\documentclass[letterpaper, 10 pt, journal, twoside]{ieeetran} 
\usepackage{mathtools}
\usepackage{psfrag}
\usepackage{amsfonts}
\usepackage{amssymb}
\usepackage{balance}
\usepackage{amsthm}
\usepackage{amsmath}
\usepackage{mathrsfs} 
\usepackage{enumerate}
\usepackage{psfrag,epsfig,graphicx,ae}
\usepackage[colorlinks=true, urlcolor=blue, citecolor=blue, linkcolor=blue]{hyperref}
 \usepackage[usenames,dvipsnames]{pstricks}
  \usepackage{pstricks-add}
 \usepackage{epsfig}
 \usepackage{pst-grad} 
 \usepackage{pst-plot} 
 \usepackage[space]{grffile} 
 \usepackage{etoolbox} 
 \makeatletter 
 \patchcmd\Gread@eps{\@inputcheck#1 }{\@inputcheck"#1"\relax}{}{}
 \makeatother
\usepackage{color}
\usepackage{subcaption}
\usepackage{url}
\usepackage{tikz}
\usetikzlibrary{patterns}
\usepackage{float}             

\newcounter{teocount}
\newcounter{propcount}

\newcounter{remcount}
\newcounter{defcount}

\newtheorem{remm}[remcount]{Remark}

\newtheorem{definition}[defcount]{Definition}
\newtheorem{proposition}[propcount]{Proposition}
\newtheorem{theorem}[teocount]{Theorem}

\newtheorem{exx}{Example}

\newtheorem{lemma}{Lemma}

\newenvironment{remark}{\begin{remm}\rm }{\hfill \hspace*{1pt} \hfill $\diamond$\end{remm}}

\newenvironment{example}{\begin{exx}\rm }{\hfill \hspace*{1pt} \hfill $\lrcorner$\end{exx}}

\DeclareMathOperator{\diag}{diag}
\DeclareMathOperator{\dom}{dom}
\DeclareMathOperator{\rg}{rge}
\DeclareMathOperator{\He}{He}

\newcommand{\np}{{n_p}}
\newcommand{\nuu}{{n_u}}
\newcommand{\1}{{\mathbf{1}}}

\newcommand{\R}{{\mathbb R}}
\newcommand{\Dy}{{\mathbb D}}
\renewcommand{\S}{{\mathbb S}}
\newcommand{\nats}{{\mathbb N}}

\DeclareMathOperator*{\argmin}{arg\,min}

\makeatletter
\newcommand*{\tr}{%
  {\mathpalette\@tr{}}%
}
\newcommand*{\@tr}[2]{%
  \raisebox{\depth}{$\m@th#1\intercal$}%
}
\makeatother

\newcommand*{\QEDB}{\null\nobreak\hfill\ensuremath{\square}}
\newcommand{\source}{{This is a repository version of our paper. Please cite the published version DOI: \href{https://doi.org/10.1109/LCSYS.2022.3183937}{https://doi.org/10.1109/LCSYS.2022.3183937}}}
 \usepackage{fancyhdr}
\makeatletter

\def\ps@IEEEtitlepagestyle{}

\title{\LARGE \bf  Stability Analysis of a Class of Discontinuous Discrete-Time Systems}
\author{Francesco Ferrante,~\IEEEmembership{Senior Member,~IEEE, } and Giorgio Valmorbida
\thanks{This research is funded in part by ANR via project HANDY, number ANR-18-CE40-0010.}
\thanks{Francesco Ferrante is with Department of Engineering, University of Perugia, Perugia, Italy. {francesco.ferrante@unipg.it}.}
\thanks{G. Valmorbida is with the Laboratoire des Signaux et Systèmes, CentraleSupélec, CNRS, Université Paris-Saclay, Gif-sur-Yvette 91192, France. He is also with Inria projet DISCO. {giorgio.valmorbida@l2s.centralesupelec.fr.}}
\thanks{The authors contributed equally to the work.}
}
\begin{document}
\maketitle
\begin{abstract}
The stability analysis of a class of discontinuous discrete-time systems is studied in this paper. The system under study is modeled as a feedback interconnection of a linear system and a set-valued nonlinearity. An equivalent representation, based on a constrained optimization problem, is proposed to represent the set-valued nonlinearity via a collection of linear and quadratic constraints. Relying on this description and on the use of a generalized quadratic set-valued Lyapunov functions, sufficient conditions in the form of linear matrix inequalities for global exponential stability are obtained. Numerical examples corroborate the theoretical findings.
\end{abstract}
\begin{IEEEkeywords}
Nonlinear systems, Lyapunov stability, LMIs.
\end{IEEEkeywords}
\section{Introduction}
\subsection{Motivation and background}
\IEEEPARstart{T}{he} widespread availability and the decreasing costs of digital devices have promoted the implementation of control systems on digital platforms. However, embedded control systems when implemented on affordable devices also raise theoretical challenges in term of stability analysis and performance.  Indeed, the presence of nonlinear elements in feedback control systems may lead to limit cycles, chaotic behaviors, which may induce poor performance and instability. A fundamental limitation in digital control systems consists of the use of finite alphabets to represent information such as inputs and outputs. The control of systems based on the use of finite alphabets has been largely investigated in the literature over the last years. A finite number of input values appears in quantized actuators followed by a saturation nonlinearity~\cite{fu2005sector,di2020practical}. For example, in \cite{cortes2006finite} the use of ternary controllers for multi agent systems consensus is proposed. Stabilization of nonlinear systems by a finite number of control or measurement values is studied in \cite{de2009robust}. Formation control under the assumption of binary information exchanges has been pursued in \cite{jafarian2015formation}. Distributed consensus via binary control has been investigated in \cite{chen2011finite}. In \cite{yu2011rendezvous}, ternary feedback controllers are shown to be effective to tackle rendez-vous problems for Dubins models of cars. Another application of ternary controllers for integrator coordination is featured in \cite{de2013robust}.  A predominant use of controllers taking values into finite alphabets arises in the literature of symbolic control in which control design is performed based on discrete abstractions; see, e.g.,  
 \cite{tabuada2009verification, SIG19}.
\subsection{Contributions and Outline of the Paper}
In this paper, we focus on stability analysis of a class of discontinuous discrete-time control systems. In particular, we consider a scenario in which a linear plant is controlled via an affine static state feedback law taking values into the set $\mathcal{Q}\coloneqq \{0, \delta_1\}\times \{0, \delta_2\}\times\dots\{0, \delta_{\nuu}\}$, where $\nuu\in\nats$ is the number of control inputs and $\delta_i\in\R$, for all $i\in\{1, 2,\dots, \nuu\}$, are some given \emph{levels}. More specifically, we focus on the following class of nonlinear discrete-time systems:
\begin{equation}
\label{eq:ternaryclosedloop_single}
x^+=Ax+B\Delta S(Kx+d)
\end{equation}
where $A\in\R^{n_p\times n_p}$, $B\in\R^{n_p\times n_u}$, $K\in\R^{n_u\times n_p}$, 
 $\Delta\coloneqq\diag\{\delta_1, \delta_2,\dots, \delta_{\nuu}\}$, $d\in\R^{\nuu}$ are given and
 $S\colon\R^{n_u}\rightarrow\R^{n_u}$ is defined as follows:
 \begin{equation}
 \label{eq:Qsingle}
S(u)\coloneqq (s(u_1), s(u_2), \dots, s(u_{\nuu}))
 \end{equation}
where for all $v\in\R$ 
 \begin{equation}
  \label{eq:step}
s(v)\coloneqq\begin{cases}
1&\text{if}\quad v>0\\
0&\text{if}\quad v\leq 0.
\end{cases}
 \end{equation}

 The above setup is rather general and allows one to capture, among others, the typical situation in which actuators may only deliver a finite set of input values.  Since system \eqref{eq:ternaryclosedloop_single} is assumed to be given, the goal of the paper is to provide a method for the stability analysis of the origin of~\eqref{eq:ternaryclosedloop_single}. The expression for the input mapping $s$ in~\eqref{eq:step} is a static nonlinearity, which is commonly studied by a sector description. In contrast with more classical absolute stability approaches, \emph{we do not rely on any sector bound approach}. Moreover, we introduce a class of set-valued piecewise quadratic Lyapunov functions (\emph{LF}), as opposed to the standard quadratic LF approaches. The structure of the paper and its contributions can be summarized as follows:
\begin{itemize}
\item Following the general approach in \cite{primbs2001kuhn}, in Section~\ref{sec:probStat} we propose an equivalent representation of a regularized version of the quantizer mapping \eqref{eq:Qsingle} based on the use of Karush-Kuhn-Tucker (\emph{KKT}) necessary conditions for optimality.
\item Inspired by \cite{dai2009piecewise,gonzaga2012stability,doi:10.1137/050629185,8392418}, in Section~\ref{sec:StabAnalysis} we introduce a suitable class of \emph{generalized quadratic} Lyapunov functions.
 
\item Relying on the proposed class of set-valued generalized-quadratic Lyapunov functions, Section~\ref{sec:StabAnalysis} ends by providing sufficient conditions in the form of linear matrix inequalities to certify global exponential stability of the origin of~\eqref{eq:ternaryclosedloop_single}. Those conditions can be efficiently checked by using semidefinite programming.
\item Section~\ref{sec:NumEx} illustrates the effectiveness of the proposed methodology in two numerical examples. 
\end{itemize}   
The main extension with respect to our  conference paper~\cite{valmorbida2020quantization} is the analysis of a set-valued regularized version of the discontinuous dynamics in \eqref{eq:ternaryclosedloop_single}. Such an extension naturally leads to the use of set-valued Lyapunov functions, which requires a proper handling; this is not pursued in \cite{valmorbida2020quantization}. 
\subsection{Notation}
\label{sec:Notation}
The symbols $\nats$ an $\R$ denote, respectively, the set of positive integers and the set of reals, $\mathbb{N}_{0}$ represents the set of nonnegative integers, $\R^n$ is the $n$-dimensional Euclidean space, and $\R^{n\times m}$ is set of the $n\times m$ real matrices. The symbol $\S^n$ stands for the set of $n\times n$ symmetric matrices, $\Dy^{n}$ denotes the set of $n\times n$ diagonal matrices, and $\mathbb{P}^{n}$ is the set of $n\times n$ symmetric matrices with nonnegative entries. For a vector $x\in\R^n$, $\vert x \vert$ denotes its Euclidean norm. 
The identity matrix is denoted by $I$. The symbol $\1_{n}$ is the all-ones vector of $\R^n$.
For a matrix $A\in\R^{n\times m}$, $A^\tr$ denotes the transpose of $A$,  and, when $n=m$, $\He (A)=A+A^\tr$. We use the equivalent notation for vectors $(x, y)=[x^\tr\,\,y^\tr]^\tr$. The symbol $\odot$ stands for the Hadamard product. Let $A\in\mathbb{S}^n$, $A\preceq 0$ stands for negative semidefiniteness of $A$. Given $x\in\mathbb{R}^n$, $x\leq 0$ indicates that the components of  $x$ are nonpositive. The symbol $\bullet$ stands for symmetric blocks in symmetric partitioned matrices. 
Given a matrix $M$ with $\ker M\neq \{0\}$, $M_\perp$ stands for any matrix having as columns a basis of $\ker M$. The symbol $\sup$ stands for the supremum, $\overline{\Omega}$ is the closure of the set $\Omega$, and $\rg f$ is the image of the function $f$.  The symbol $\bigoplus_{i=1}^n A_i$ stands for direct sum of matrices $A_1, A_2, \dots, A_n$ and $A \otimes B$ indicates the Kronecker product of matrices $A$ and $B$.
\section{Problem setting and Key results}
\label{sec:probStat}
\subsection{Modeling and structural properties}
\label{sec:setup}
Due to the discontinuity of $S$ at zero, \eqref{eq:ternaryclosedloop_single} is a discontinuous dynamical system. Although discontinuities in discrete-time dynamical systems do not lead to major technical problems as in their continuous-time counterpart (see, e.g., \cite{Cortes:2008aa,cer:dep:fra/automatica2011,ferrante2019sensor,ferrante2015stabilization}), they generally lead to lack of robustness, with stability properties being fragile in the presence of vanishing perturbations; see \cite[Example 4.4, page 76]{goebel2012hybrid}. To overcome this drawback, in this work we consider the following set-valued \emph{regularization} of \eqref{eq:ternaryclosedloop_single}:
\begin{equation}
\label{eq:ternaryclosedloop}
x^+\in Ax+B\Delta\mathbf{S}(Kx+d)
\end{equation}
where the set-valued mapping\footnote{The double arrow notation $\rightrightarrows$ is used to distinguish set-valued maps from functions.} 
 $\mathbf{S}\colon\R^{n_u}\rightrightarrows\R^{n_u}$ is defined as follows:
 \begin{equation}
 \label{eq:fraQ}
 \mathbf{S}(u)\coloneqq (\mathbf{s}(u_1),\mathbf{s}(u_2), \dots, \mathbf{s}(u_{\nuu}))
 \end{equation}
with, for all $v\in\R$, 
\begin{equation}
\mathbf{s}(v)\coloneqq\begin{cases}
1&\text{if}\quad v>0\\
0&\text{if}\quad v<0\\
[0, 1]&\text{if}\quad v=0.
\end{cases}
\label{eq:stepSet}
\end{equation}
Observe that solutions to \eqref{eq:ternaryclosedloop_single}  are solutions to \eqref{eq:ternaryclosedloop}. Thus, stability properties of \eqref{eq:ternaryclosedloop} carry over \eqref{eq:ternaryclosedloop_single}. We discuss properties of solutions and provide stability definitions to difference inclusions in Section~\ref{sec:prel}. 
\begin{remark}
It can be shown that $\mathbf{S}$ contains the so-called (discrete-time) Krasovskii regularization of the step function $S$; see, e.g., \cite[Definition 4.13]{goebel2012hybrid}. Therefore, \eqref{eq:ternaryclosedloop} captures all possible solutions to \eqref{eq:ternaryclosedloop_single} obtained by introducing vanishing state perturbations, i.e., Hermes solutions; see \cite[Chapter 4]{goebel2012hybrid}. This ensures that stability of the origin of \eqref{eq:ternaryclosedloop_single} is robust with respect to vanishing perturbations. 
\end{remark}
\subsection{Characterization of the mapping $\mathbf{S}$ via quadratic constraints}
\label{sec:KKT}
In this subsection we illustrate the key result of this paper. This result yields a tight characterization of the mapping $\mathbf{S}$ in \eqref{eq:fraQ} in terms of quadratic constraints. To achieve this goal, we pursue a similar approach as in \cite{primbs2001kuhn} and rely on optimization-based representation of the mapping $\mathbf{S}$ along with Karush-Kuhn-Tucker (KKT) optimality conditions. 
To this end, observe that for all $v\in\R$, one can express \eqref{eq:stepSet} as
\begin{equation}
\label{eq:OptiStep}
\begin{array}{rl}
\mathbf{s}(v)\in\underset{w\in[0,1]}{\argmin}& -vw.
\end{array}
\end{equation}
Clearly, if $v<0$, one has $\mathbf{s}(v)=0$, if $v>0$, one has $\mathbf{s}(v)=1$, while when $v=0$, $\mathbf{s}(v)\in [0,1]$, which is consistent with \eqref{eq:stepSet}. Building upon this observation, one can obtain a characterization of the mapping $\mathbf{S}$ via the application of Karush-Kuhn-Tucker (KKT) optimality conditions to problem \eqref{eq:OptiStep}. This is formally stated in the result given next. 
\begin{theorem}
\label{thm:stepKKT}
Let $\mathbf{S}$ be defined as in \eqref{eq:fraQ}, $u\in\R^{n_u}$, and $s\in \R^{n_u}$. Then, the following items are equivalent:
\begin{itemize}
\item[$(i)$] $s\in \mathbf{S}(u)$
\item[$(ii)$] there exist $\lambda_1, \lambda_2\in\R^{n_u}$ such that:
\begin{subequations}
\label{eq:KKT_F} 
\begin{align}
\label{eq:KKT_Fa} 
- u - \lambda_1+  \lambda_2 &  =   0\\
\label{eq:KKT_Fb} 
\lambda_1  \odot s&  =   0 \\
\label{eq:KKT_Fc} 
\lambda_2\odot(\1_{\nuu}-s)&  =   0 \\
\label{eq:KKT_Fd} 
-\lambda_1&   \leq    0 \\
\label{eq:KKT_Fe} 
-\lambda_2&   \leq    0\\
\label{eq:KKT_Ff} 
-s&   \leq  0 \\
\label{eq:KKT_Fg} 
-\1_{\nuu}+s&  \leq  0
\end{align}
\end{subequations}
\end{itemize}
\end{theorem}
\begin{proof}
Since the relations in \eqref{eq:KKT_F}  are defined elementwise, the claim can be proven for each element. Thus, we assume $n_u =1$, in which case we have $\bf{s} = \bf{S}$.

\noindent\underline{\textit{Proof of $(i)\implies (ii)$}}. Using \eqref{eq:OptiStep}, it follows that
\begin{equation}
\label{eq:OptiS}
s\in\argmin_{w\in[0,1]} -uw
\end{equation}
In particular, the Lagrangian associated to \eqref{eq:OptiS} writes:
$$\mathcal{L}_u(w,\lambda) = -uw+\left[\begin{array}{c}\lambda_1\\\lambda_2\end{array}\right]^{\top}  \left( \left[\begin{array}{c}-1\\1\end{array}\right]w+\left[\begin{array}{c}0\\-1\end{array}\right] \right).$$
To conclude, let us recall that from KKT necessary conditions for optimality (\eqref{eq:OptiS} is a linear program), one has that for any optimal solution $w^\star$ to \eqref{eq:OptiS}, there exists a unique $\lambda^\star\coloneqq(\lambda_1^\star, \lambda_2^\star)$ such that:
\begin{equation}
\label{eq:KKTns}
\begin{aligned}
&\frac{d}{dw}\mathcal{L}_u(w^\star, \lambda^\star)=0, \lambda^\star_1w^\star=0, \lambda^\star_2(1-w^\star)=0\\
&\lambda_1^\star\geq 0, \lambda_2^\star\geq 0, w^\star\geq 0, w^\star\leq 1\\
\end{aligned}
\end{equation} 
which reads as \eqref{eq:KKT_F}. Hence, recalling that $s$ is an optimal solution to \eqref{eq:OptiS}, i.e., $w^\star=s$ satisfies \eqref{eq:KKTns}, the implication is established. 

\noindent\underline{\textit{Proof of $(ii)\implies (i)$}}. This implication can be readily shown by observing that since \eqref{eq:OptiS} is a linear program, the satisfaction of \eqref{eq:KKT_F} (KKT conditions) implies \eqref{eq:OptiS}. This establishes the result.
\end{proof}

Theorem~\ref{thm:stepKKT} shows that for all $u\in\R^{\nuu}$ and $s\in\mathbf{S}(u)$, there exist $\lambda_1, \lambda_2\in\R^{\nuu}$ such that $\chi\coloneqq(\lambda_1, \lambda_2, s, \1_{\nuu}-s, u)\in\R^{5\nuu}$ satisfies \eqref{eq:KKT_F}. In particular, the entries of the vector $\chi$ depend on $u$ and $s$. Therefore, in the remainder of the paper, given $u\in\R^{\nuu}$ and $s\in\mathbf{S}(u)$, we use the shorthand notation $\chi(u, s)$ to denote the corresponding vector satisfying \eqref{eq:KKT_F}. For compactness, next we rewrite the linear equality constraints in \eqref{eq:KKT_Fa} as follows:
\begin{subequations}
\begin{equation}
L\chi=0
\label{eq:KKTCompact}
\end{equation}
where:
\begin{equation}
\begin{aligned}
&L\coloneqq \begin{bmatrix}
-1&1&0&0&-1
\end{bmatrix}\otimes I_{\nuu}.
\end{aligned}
\label{eq:ConstraintsMatrices}
\end{equation}
\end{subequations}

The result given next provides an explicit characterization of the multipliers $\lambda_1$ and $\lambda_2$ introduced in Theorem~\ref{thm:stepKKT}. This characterization enables to make the construction of the Lyapunov in Section~\ref{sec:StabAnalysis} explicit.
\begin{lemma}
\label{lemm:ExplLag}
Let $u\in\R^{\nuu}$. Then, for all $s\in\mathbf{S}(u)$, there exist $\lambda_1(u)$ and $\lambda_2(u)$ such that  $\chi(u, s)=(\lambda_1(u), \lambda_2(u), s, \1_{\nuu}-s, u)$ satisfies \eqref{eq:KKT_F}. In particular
\begin{equation}
\label{eq:lemm:ramp}
\begin{aligned}
\lambda_1(u)=r(-u),\quad&\lambda_2(u)=r(u),
\end{aligned}
\end{equation}
where $u\mapsto r(u)$ is the componentwise ramp function, namely for all $i=1,2,\dots, \nuu$, $r(u_i)=u_i$ if $u_i\geq 0$ and $r(u_i)=0$ otherwise. \QEDB
\end{lemma} 
\begin{proof}
For the sake of the exposition, we develop the proof for $\nuu=1$. To prove the result, we analyze the solutions to system \eqref{eq:KKT_F}  in the unknowns $(\lambda_1, \lambda_2)$ for fixed values of $u$ and $s$.  In particular, the following can be proven via simple manipulations. If $u=0$, then $s\in[0, 1]$ and from \eqref{eq:KKT_Fb}, \eqref{eq:KKT_Fc}, $\lambda_1=\lambda_2=0$. If $u\neq 0$. Then, $s\in\{0, 1\}$ and from \eqref{eq:KKT_Fa}, $(\lambda_1, \lambda_2)=(-u, 0)$ if $s=0$ or $(\lambda_1, \lambda_2)=(0, u)$ otherwise. 
Hence, using the definition of the map $\mathbf{S}$, the two relationships above yield \eqref{eq:lemm:ramp}.
\end{proof}

In light of Lemma~\ref{lemm:ExplLag}, in the remainder of the paper, for all $u\in\R^{\nuu}$ we use the notation $\overline{\lambda}(u)=(r(-u), r(u))$.

\subsection{Preliminaries on difference inclusions}
\label{sec:prel}
We consider set-valued dynamics of the form:
\begin{equation}
\label{eq:DiffInc}
x^+\in G(x)
\end{equation}
where $x\in\R^n$ is the system state and $G\colon \R^n\rightrightarrows \R^n$ is a set-valued map. A solution to \eqref{eq:DiffInc} is any function $\phi\colon\dom\phi\rightarrow \R^n$ with $\dom\phi=\nats_0\cap\{0, 1,\dots, \overline{J}\}$ for some $\overline{J}\in\nats_0\cup\{\infty\}$ such that for all $j\in\dom\phi$, with $j+1\in\dom\phi$, $\phi(j+1)\in G(\phi(j))$.  We say that a solution $\phi$ is \emph{maximal} if it cannot be extended and it is \emph{complete} if $\sup\dom\phi=\infty$. Regarding system~\eqref{eq:ternaryclosedloop}, the following holds
\begin{proposition}
\label{prop:existence}
For any $\xi\in\R^{\np}$, there exists a  complete solution $\phi$ to \eqref{eq:ternaryclosedloop} such that $\phi(0)=\xi$.
 \end{proposition}
\begin{proof}
The proof follows simply from the fact that $\mathbf{S}$ is defined everywhere; see, e.g., \cite[Proposition 2.10]{goebel2012hybrid}.  
\end{proof}

The following notion of global exponential stability is used in the paper. 
\begin{definition}
We say the the origin is globally exponentially stable (\emph{GES}) for \eqref{eq:DiffInc} if there exists $\lambda, \kappa>0$ such that any maximal solution $\phi$ to \eqref{eq:DiffInc} satisfies, for all $j\in\dom\phi$,
$
\vert \phi(j)\vert\leq \kappa e^{-\lambda j}\vert \phi(0)\vert.
$
\hfill$\diamond$
\end{definition}

Next we provide sufficient conditions for GES
of the origin of  \eqref{eq:DiffInc}. Those conditions are formulated 
in terms of Lyapunov inequalities involving a set-valued Lyapunov function.
\begin{theorem}
\label{thm:Lyapunov}
Suppose that there exists $V\colon \R^n\rightrightarrows \R$, and positive real numbers $c_1, c_2, c_3$, and $p$ such that 
\begin{align}
\label{eq:Sandwhich}
&c_1\vert x\vert^p\leq \sup V(x)\leq c_2\vert x\vert^p, ~\forall x\in\R^n,\\
\label{eq:DeltaV}
&\sup V(g)-\sup V(x)\leq -c_3\vert x\vert^p, ~\forall x\in\R^n, g\in G(x).
\end{align}
Then, the origin is GES for  \eqref{eq:DiffInc}. \QEDB
\end{theorem}
\begin{proof}
For all $x\in\R^n$, define $W(x)\coloneqq\sup V(x)$. The proof of the statement follows directly by observing that $W$ is a standard single-valued Lyapunov function for \eqref{eq:DiffInc}. \end{proof}

The conditions given in Theorem~\ref{thm:Lyapunov} are in general difficult to check. To overcome this drawback, we provide the following result that is easier to exploit.
\begin{proposition}
\label{cor:DeltaVuff}
Let $V\colon \R^n\rightrightarrows \R$, and $c_1, c_2, c_3$, and $p$ as in  Theorem~\ref{thm:Lyapunov}. We have the following:
\begin{itemize}
\item[$(i)$] if
\begin{equation}
\label{eq:suffSandwhich}
\begin{aligned}
&c_1\vert x\vert^p\leq \varrho\leq c_2\vert x\vert^p,&\forall x\in\R^n, \varrho\in\overline{\rg V(x)}
\end{aligned}
\end{equation}
then \eqref{eq:Sandwhich}  holds;
\item[$(ii)$] Let $x\in\R^n$ and $g\in G(x)$. If there exists $\omega\in V(x)$ such that: 
\begin{equation}
\label{eq:suffDeltaV}
\psi-\omega\leq -c_3\vert x\vert^p,\quad \forall \psi\in \overline{V(g)}
\end{equation}
then, \eqref{eq:DeltaV} holds.
\end{itemize}
\end{proposition}
\begin{proof}
From \eqref{eq:suffSandwhich}, it follows that for all $x\in\R^{n}$, $V(x)$ is bounded. This implies that for all $x\in\R^{n}$, $\sup V(x)\in\overline{\rg V(x)}$. Hence, item $(i)$ is established. The proof of item $(ii)$ easily follows from the fact that since
$\sup V(g)-\sup V(x)\leq \sup V(g)-\omega$
and $V(g)$ is bounded, there exists $\psi^\star\in\overline{V(g)}$ such that $\sup V(g)=\psi^\star$.
\end{proof}
\section{Stability Analysis}
\label{sec:StabAnalysis}
We are now in a position to state the main result of this paper. This result provides sufficient conditions for global exponential stability of \eqref{eq:ternaryclosedloop} in the form of matrix inequalities. To this end, we use Theorem~\ref{thm:Lyapunov} and Proposition~\ref{cor:DeltaVuff}, and the following set-valued Lyapunov function candidate:
\begin{equation}
\label{eq:LyapVSet}
V(x)\coloneqq\bigcup_{s\in\mathbf{S}(Kx+d)}\left\{\begin{bmatrix}
x\\
s\\
\overline{\lambda}(Kx+d)
\end{bmatrix}^\tr P \begin{bmatrix}
x\\
s\\
\overline{\lambda}(Kx+d)
\end{bmatrix}\right\}
\end{equation}
A prototype of the function \eqref{eq:LyapVSet} for the scalar case with 
$$
P=\left[\begin{smallmatrix}
1 &0&0&0\\
\bullet&1&0 &0\\
\bullet&\bullet&0 &1\\
\bullet&\bullet&\bullet &0
\end{smallmatrix}\right], K=1, d=-1$$ 
is depicted in \figurename~\ref{fig:Lyap}.
\begin{figure}
\centering
\medskip

\psfrag{x}[1][1][1]{$x$}
\psfrag{y}[1][1][1]{$V(x)$}
\includegraphics[trim=0.3cm 0 1cm 0.4cm, clip, width=0.9\columnwidth]{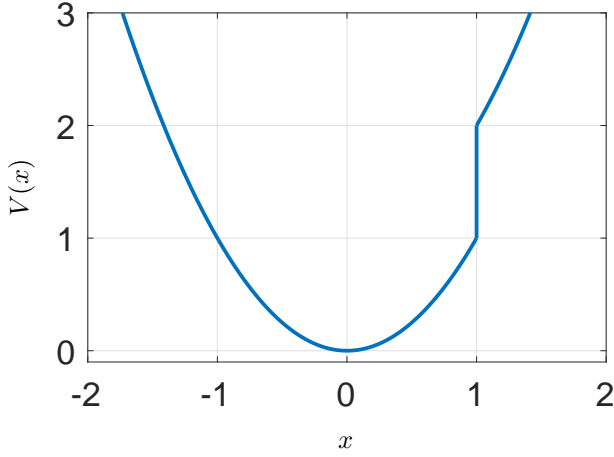}
\caption{Prototype of the Lyapunov function candidate used in Theorem~\ref{thm:StabCondLMIs}.\label{fig:Lyap}}
\end{figure}
\begin{theorem}
\label{thm:StabCondLMIs}
Suppose that there exist $P\in\mathbb{S}^{\np+3\nuu}$, $M_1, M_2, M_3\in\mathbb{P}^{8\nuu+1}$,
$\widehat{G}_i=(G_{i, 1}, G_{i, 2})\in\mathbb{D}^{2\nuu}\times \mathbb{D}^{2\nuu}$, $c_i>0$, with $i\in\{1, 2, 3\}$  such that
\begin{subequations}
\label{eq:StabCond}
\begin{equation}
\label{eq:SandV}
\begin{aligned}
&W^\tr_\perp(V_u+T^\tr \Psi(\widehat{G}_1)T+F^\tr M_1 F)W_\perp\preceq 0\\
&W^\tr_\perp(V_l+T^\tr \Psi(\widehat{G}_2) T+F^\tr M_2 F)W_\perp\preceq 0 
\end{aligned}
\end{equation}
\begin{equation}
W^\tr_\perp(\Xi+T^\tr \Psi(\widehat{G}_3)T+F^\tr M_3 F)W_\perp\preceq 0
\label{eq:DeltaV_fins}
\end{equation}
\end{subequations}
where for all $i\in\{1, 2, 3\}$:
\begin{subequations}
\begin{equation}
\label{eq:dataMainTh}
\begin{aligned}
&\Psi(\widehat{G}_i)\coloneqq\He\left(\bigoplus_{j=1}^2
\begin{bmatrix}
0_{2\nuu, 2\nuu}&G_{i, j}&0_{2\nuu, \nuu}\\
0_{3\nuu, 2\nuu}&0_{3\nuu, 2\nuu}&0_{3\nuu,\nuu}
\end{bmatrix}\right)\\
&\Xi\coloneqq \begin{bmatrix}\Pi_1^\tr \mathscr{V}_+^\tr&\Pi_2^\tr\end{bmatrix}\left(\begin{bmatrix}
1&0\\
0&-1
\end{bmatrix}\otimes P\right)
\begin{bmatrix}\mathscr{V}_+\Pi_1\\\Pi_2\end{bmatrix}+c_3 X\\
&V_l\coloneqq -\Pi_2^\tr P\Pi_2+c_1 X, V_u\coloneqq\Pi_2^\tr P\Pi_2-c_2 X,
\end{aligned}
\end{equation}
where $\Pi_1$, $\mathscr{V}_+$, and $F$ are defined in \eqref{eq:Pi1},
\begin{figure*}
\begin{equation}
\label{eq:Pi1}
\begin{aligned}
&\mathscr{V}_+\coloneqq\begin{bmatrix}
A&B\Delta&0_{\np,\nuu}&0_{\np,2\nuu}&0_{\np,2\nuu}\\
0_{\nuu,\np}&0_{\nuu,\nuu}&I_{\nuu}&0_{\nuu,2\nuu}&0_{\nuu,2\nuu}\\
0_{2\nuu,\np}&0_{2\nuu,\nuu}&0_{2\nuu,\nuu}&0_{2\nuu,2\nuu}&I_{2\nuu}
\end{bmatrix}, &F=\begin{bmatrix}
0_{4\nuu,\np}&\begin{bmatrix}I_{4\nuu}&0_{4\nuu,\nuu}\end{bmatrix} &0_{4\nuu,5\nuu}&0_{4\nuu,1}\\
0_{4\nuu,\np}&0_{4\nuu,5\nuu}&\begin{bmatrix}I_{4\nuu}&0_{4\nuu,\nuu}\end{bmatrix}&0_{4\nuu,1}\\
0_{1,\np}&0_{1,5\nuu}&0_{1, 5\nuu}&1
\end{bmatrix}\\
&\Pi_1\coloneqq \begin{bmatrix}
I_{\np}&0_{\np, 5\nuu}&0_{\np, 5\nuu+1}\\
0_{\nuu,\np}&H&0_{\nuu, 5\nuu+1 }\\
0_{\nuu,\np}&0_{\nuu,5\nuu}&\begin{bmatrix}H&0_{\nuu, 1}\end{bmatrix}\\
0_{2\nuu,\np}&\begin{bmatrix} I_{2\nuu}&0_{2\nuu, 3\nuu}\end{bmatrix}&0_{2\nuu, 5\nuu+1}\\
0_{2\nuu,\np}&0_{2\nuu, 5\nuu}&\begin{bmatrix} I_{2\nuu}&0_{2\nuu, 3\nuu+1}\end{bmatrix}
\end{bmatrix}
\end{aligned}
\end{equation}
\end{figure*}
\begin{equation}
\begin{aligned}
&H\coloneqq\begin{bmatrix}
0_{\nuu,2\nuu}&I_{\nuu}&0_{\nuu,2\nuu}
\end{bmatrix}, X\coloneqq I_{\np}\oplus 0_{10\nuu+1, 10\nuu+1}\\
&\Pi_2\coloneqq\begin{bmatrix}
I_{\np+\nuu}&0_{\np+\nuu,5\nuu}\\
0_{2\nuu,(\np+\nuu)}&\begin{bmatrix} 0_{2\nuu,\nuu}& I_{2\nuu}&0_{2\nuu,2\nuu}\end{bmatrix}
\end{bmatrix}\Pi_1\\
&W\coloneqq\begin{bmatrix}
R\\
(I_2\otimes L)T
\end{bmatrix}, T\coloneqq\begin{bmatrix}
0_{10\nuu,\np}&I_{10\nuu}&0_{10\nuu,1}
\end{bmatrix}\\
&R\coloneqq\begin{bmatrix}
K&-Z&0_{\nuu, 5\nuu}&d\\
KA&KB\Delta J&-Z&d\\
0_{\nuu, n_p}&E&0_{\nuu, n_p}&-\1\\
0_{\nuu, n_p}&0_{\nuu, n_p}&E&-\1
\end{bmatrix}\\
&E\coloneqq\begin{bmatrix}
0_{\nuu, 2\nuu}&I_{\nuu}&I_{\nuu}&0_{\nuu,\nuu}
\end{bmatrix}\\ 
&Z\coloneqq\begin{bmatrix}
0_{\nuu, 4\nuu}&I_{\nuu}
\end{bmatrix}, J\coloneqq\begin{bmatrix}
0_{\nuu, 2\nuu}&I_{\nuu}&0_{\nuu, 2\nuu}
\end{bmatrix},
\end{aligned}
\end{equation}
\end{subequations}
and $L$ is defined in \eqref{eq:ConstraintsMatrices}. Then, the origin of \eqref{eq:ternaryclosedloop} is GES.
\end{theorem}
\begin{proof}
The proof hinges upon Theorem~\ref{thm:Lyapunov} and Proposition~\ref{cor:DeltaVuff}. In particular, let, for all $x\in\R^{\np}$, $V$ be defined as in \eqref{eq:LyapVSet}. 
We show that the satisfaction of \eqref{eq:StabCond}  implies all the conditions in  Proposition~\ref{cor:DeltaVuff}.
Pick $x\in\R^{n_p}$ and $g\coloneqq Ax+B\Delta s\in Ax+B\Delta\mathbf{S}(Kx+d)$. Let 
$$
\omega=\begin{bmatrix}
x\\
s\\
\overline{\lambda}(Kx+d)
\end{bmatrix}^\tr P \begin{bmatrix}
x\\
s\\
\overline{\lambda}(Kx+d)
\end{bmatrix}
$$
and observe that $\omega\in V(x)$.  Pick any $\psi\in V(g)$. In particular, $\psi$ writes as
$$
\psi=h(s_\psi)^\tr P \underbrace{\begin{bmatrix}
Ax+B\Delta s\\
s_\psi\\
\overline{\lambda}(K(Ax+B\Delta s)+d)
\end{bmatrix}}_{h(s_\psi)}
$$
for some $s_\psi\in \mathbf{S}(K(Ax+B\Delta s)+d)$.

\underline{\textit{Preliminary steps}}. Define
$$
\theta\coloneqq\left(
x,
\chi(Kx+d, s),
\chi(K(Ax+B\Delta s)+d, s_\psi\right),
1)\in\R^{n_\theta}
$$
with $n_\theta\coloneqq \np+10\nuu+1$. Then, by construction, one has 
\begin{subequations}
\begin{equation}
\label{eq:omegax}
(x, s, \overline{\lambda}(Kx+d))=\Pi_2\theta
\end{equation}
\begin{equation}
\label{eq:omegaxg}
(x, s,
s_\psi, \overline{\lambda}(Kx+d), \overline{\lambda}(K(Ax+B\Delta s)+d)
)=\Pi_1\theta
\end{equation}
\end{subequations}
and  $h(s_\psi)=\mathscr{V}_+\Pi_1\theta$. In particular
\begin{equation}
\label{eq:DeltaV_quad}
\psi-\omega+c_3\vert x\vert^2=\theta^\tr\Xi\theta.
\end{equation}
Now observe that from the definition of $\theta$ and the general structure of the vector $\chi$, the following holds:
\begin{equation}
\label{eq:R}
R\theta=0.
\end{equation}
Moreover, using the constraints provided by Theorem~\ref{thm:stepKKT}, it follows that:
\begin{subequations}
\begin{align}
\label{eq:LTFins}
&(I_2\otimes L)T\theta=0\\
\label{eq:LTQuadEq}
&\theta^\tr T^\tr\Psi(\widehat{G}_i) T\theta=0\quad \forall i\in\{1,2,3\}\\
\label{eq:LTQuadIn}
&\theta^\tr F^\tr M_i F\theta\geq 0\qquad \forall i\in\{1,2,3\},
\end{align}
\end{subequations}
where \eqref{eq:LTQuadIn} comes from the nonnegativity constraints in Theorem~\ref{thm:stepKKT} that ensure $-F\theta\leq 0$. In particular, combining \eqref{eq:R} and \eqref{eq:LTFins} yields:
\begin{equation}
\label{eq:FinalFinsler}
W\theta=0.
\end{equation}
\underline{\textit{Proof of \eqref{eq:SandV}$\implies$\eqref{eq:suffSandwhich}.}} 
Bearing in mind \eqref{eq:omegax}, the satisfaction of \eqref{eq:suffSandwhich} is equivalent to
\begin{equation}
\label{eq:VuVl}
\begin{aligned}
&\theta^\tr V_u\theta\leq 0, &\theta^\tr V_l\theta\leq 0.
\end{aligned}
\end{equation}
Using \eqref{eq:LTQuadEq} and \eqref{eq:LTQuadIn}
\begin{equation}
\label{eq:VuVlBound}
\begin{array}{ccc}
&\theta^\tr V_u\theta\leq \theta^\tr (V_u+T^\tr\Psi(\widehat{G}_1) T+F^\tr M_1 F)\theta\\
&\theta^\tr V_l\theta\leq \theta^\tr (V_l+T^\tr\Psi(\widehat{G}_2) T+F^\tr M_2 F)\theta
\end{array}
\end{equation}
Therefore, combining \eqref{eq:VuVlBound} and \eqref{eq:FinalFinsler}, to show the satisfaction of \eqref{eq:VuVl} is enough to show that the following implication holds:
 \begin{equation}
\label{eq:Impli}
W\theta=0\implies\left\{\begin{aligned}
&\theta^\tr (V_u+T^\tr\Psi(\widehat{G}_1)T+F^\tr M_1 F)\theta\leq0\\
&\theta^\tr (V_l+T^\tr\Psi(\widehat{G}_2)T+F^\tr M_2 F)\theta\leq 0
\end{aligned}\right.
\end{equation}
The latter is equivalent to \eqref{eq:SandV}. Hence, \eqref{eq:SandV} implies \eqref{eq:suffDeltaV}.

\underline{\textit{Proof of \eqref{eq:DeltaV_fins}$\implies$\eqref{eq:suffDeltaV}.}} 
Recalling \eqref{eq:DeltaV_quad}, \eqref{eq:suffDeltaV} holds if $\theta^\tr\Xi\theta\leq0$.
Using \eqref{eq:LTQuadEq} and \eqref{eq:LTQuadIn}
\begin{equation}
\label{eq:OmegaMBound}
\begin{aligned}
\theta^\tr \Xi\theta\leq\theta^\tr (\Xi+T^\tr\Psi(\widehat{G}_3)T+F^\tr M_3 F)\theta
\end{aligned}
\end{equation}
Therefore, by recalling \eqref{eq:FinalFinsler}, the satisfaction of \eqref{eq:suffDeltaV} 
follows from the following implication
$
W\theta =0\implies  \theta^\tr (\Xi+T^\tr\Psi(\widehat{G}_3)T+F^\tr M_3 F)\theta
\leq 0,$
which in turn is equivalent to \eqref{eq:DeltaV_fins}. Namely, \eqref{eq:DeltaV_fins} implies \eqref{eq:suffDeltaV} and this concludes the proof.
\end{proof}
\begin{remark}
\label{rem:PositiveP}
The satisfaction of \eqref{eq:SandV}, which in turn ensures that \eqref{eq:suffSandwhich} holds for the function \eqref{eq:LyapVSet}, does not imply that the matrix $P$ is positive definite. This consideration clearly emerges in the numerical examples presented in Section~\ref{sec:NumEx}.
\end{remark}
\section{Numerical Examples}
\label{sec:NumEx}
In this section, we showcase the proposed methodology in two numerical examples. Specifically, we consider the following dynamical system\footnote{Numerical solutions to LMIs are obtained in YALMIP \cite{yalmip} via SeDuMi \cite{sturm1999using}.} 
\begin{equation}
\label{eq:DT_plant}
x^+=Ax+B^\prime\varphi(x)
\end{equation}
where
$A=\left[\begin{array}{cc} 0.9464 & 0.0957\\ -0.9568 & 0.9033 \end{array}\right]$
and $B^\prime$ and $\varphi\colon\R^{2}\rightarrow\R$ are selected in each of the examples below.  

\begin{example}[Ternary Control]
\label{ex:Ex1}
In this first example, we pick  $B^\prime=\begin{bmatrix}0.0049 &0.0959 \end{bmatrix}^\tr$, $K^\prime=\left[\begin{array}{cc} 9.9 &0.495
\end{array}\right]$, and analyze the case of ternary control systems; see, e.g., \cite{de2013robust,valmorbida2020quantization}. More specifically, we select
$\varphi(x)=\tau(K^\prime x)$,  where: $\tau(u)=1$ if $u>1$, $\tau(u)=0$ if $u\in[-1, 1]$, and $\tau(u)=-1$ if $u<-1$. It is worth to observe that no common quadratic function exists to certify exponential stability of the matrices $A$ and $A+B^\prime K^\prime$. This prevents from using a quadratic Lyapunov function to certify global exponential stability in this example.
The proposed methodology instead enables to certify GES. System \eqref{eq:DT_plant} can be rewritten as \eqref{eq:ternaryclosedloop_single} by taking
$
B=\begin{bmatrix}
B^\prime&-B^\prime
\end{bmatrix}$, $K=\begin{bmatrix}
K^\prime\\-K^\prime
\end{bmatrix}, d=-\1_2$, and $\Delta=I_2$. By solving the conditions in \eqref{eq:StabCond}, we obtain:
$$
P=\scalebox{0.87}{$\left[\begin{smallmatrix}
150 & 7.1 & 0.57 & 30 & -0.031 & 6.5\times 10^{-3} & 0.012 & 0\\ 
\bullet & 16 & 0.054 & 1.5 & -1.6\times 10^{-3} & 3.2\times 10^{-4} & 5.8\times 10^{-4} & 0\\ 
\bullet & \bullet  & -0.053 & 0.053 & 1.8\times 10^{-4} & 1.8\times 10^{-4} & -1.6 & -1.8\times 10^{-4}\\ 
\bullet  & \bullet  & \bullet  & -0.053 & 1.6 & -6.4\times 10^{-4} & -1.6 & -6.4\times 10^{-4}\\ 
\bullet  & \bullet  & \bullet  & \bullet  & -2.3\times 10^{-3} & 1\times 10^{-3} & 1.5\times 10^{-3} & 0\\ 
\bullet &\bullet & \bullet  &\bullet & \bullet  & 2.3\times 10^{-4} & 1.4\times 10^{-4} & 1.4\times 10^{-4}\\ 
\bullet  &\bullet  &\bullet  & \bullet  & \bullet  & \bullet  & -6.5\times 10^{-4} & -1.2\times 10^{-3}\\ 
\bullet & \bullet  & \bullet  & \bullet & \bullet  &\bullet  & \bullet & -5.1\times 10^{-4} 
\end{smallmatrix}
\right]$}.
$$
Note that, as mentioned in Remark~\ref{rem:PositiveP}, in this example the matrix $P$ is not positive definite. 
\figurename~\ref{fig:LyapLivEx1} depicts level sets of the corresponding function along with a trajectory of the system. The evolution of the values of $W$ is also presented for the same trajectory.
\begin{figure}
\centering
\psfrag{x1}[1][1][1]{$x_1$}
\psfrag{x2}[1][1][1]{$x_2$}
\psfrag{x}[1][1][1]{$j$}
\psfrag{y}[1][1][0.8]{$W(\phi(j))$}
\includegraphics[trim=0cm 0cm 1cm 0.3cm, clip, width=0.95\columnwidth]{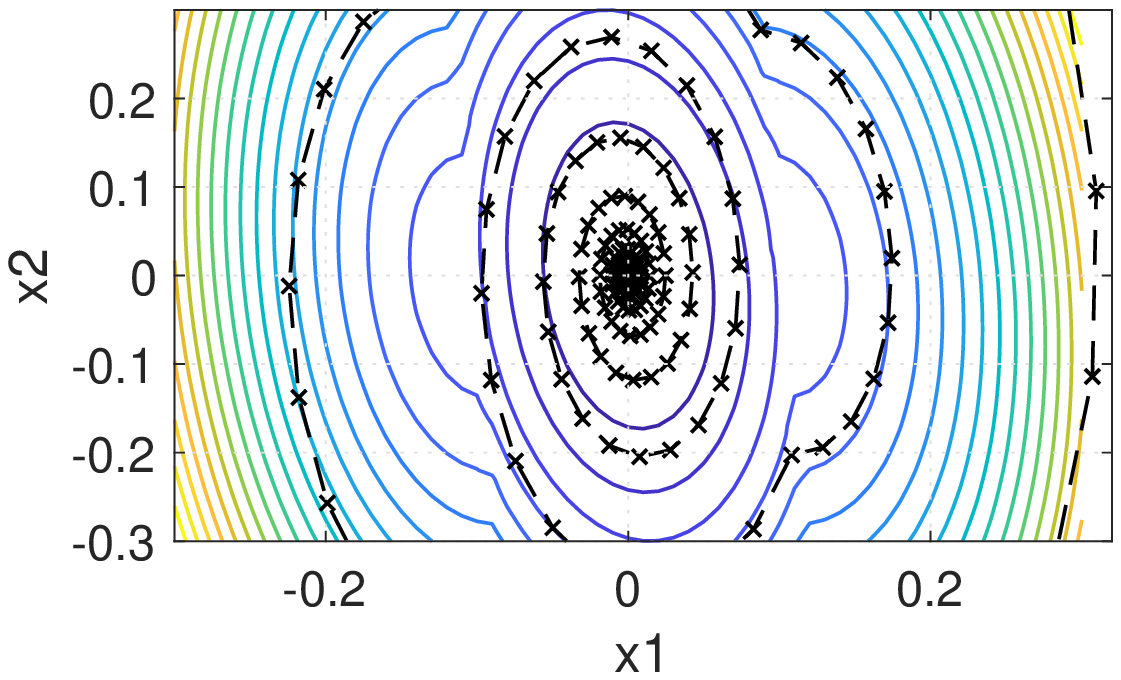}
\includegraphics[trim=0cm 0cm 0.7cm 0.1cm, clip, width=0.95\columnwidth]{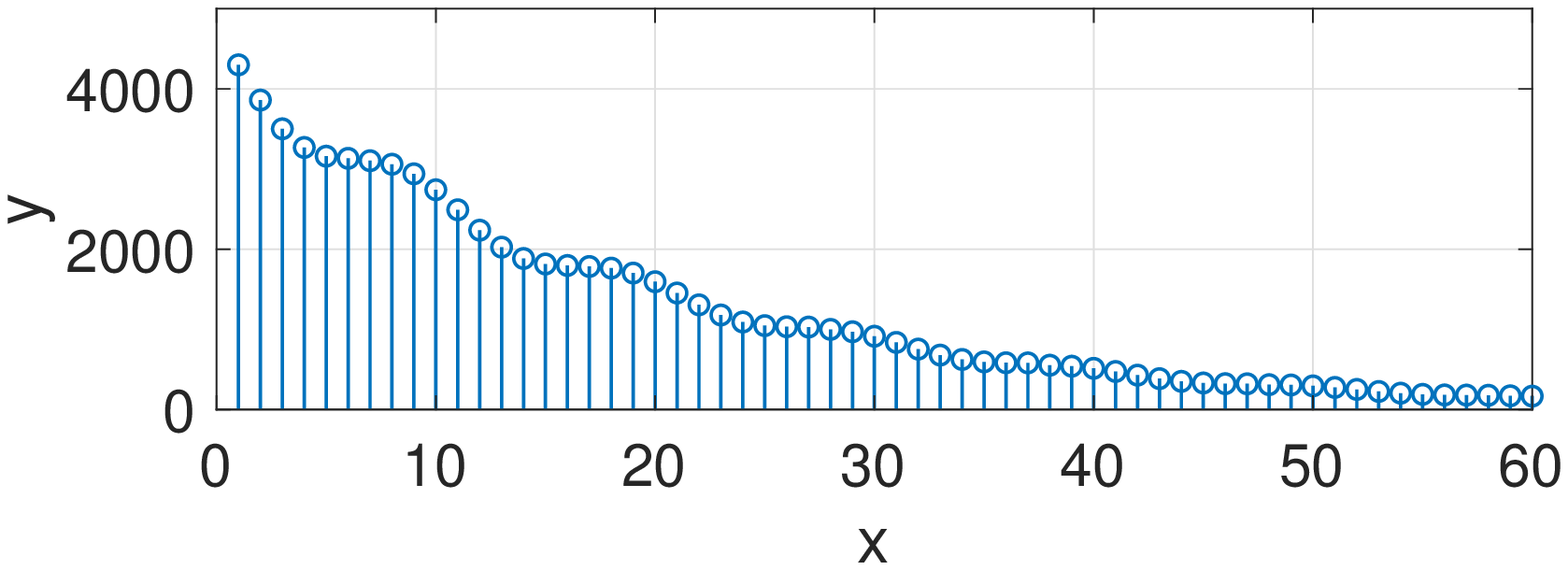}
\caption{Simulations in Example~\ref{ex:Ex1}. Top: Level sets of the function $W(x)=\sup V(x)$ and the trajectory $\phi$ of \eqref{eq:DT_plant} starting from $(5, 5)$ (dashed-crossed line). Bottom: Evolution of $j\mapsto W(\phi(j))$.
\label{fig:LyapLivEx1}}
\end{figure}
\end{example}
\begin{example}[Binary Control]
\label{ex:Ex2}
In this second example, we take $B^\prime$ and $K^\prime$ as in Example~\ref{ex:Ex1} and $\varphi(x)=s(K^\prime x-1)$. Solving the conditions in \eqref{eq:StabCond} in this case yields:
$$
P=\left[\begin{smallmatrix} 
140& 6.9 & 0.69 & -8.1\times 10^{-5} & 0.23\\ 
\bullet & 16.0 & 0.063 & 8.6\times 10^{-5} & 0.011\\ 
\bullet& \bullet & -0.066 & 9.4\times 10^{-3} & -1.5\\ 
\bullet & \bullet& \bullet & 0 & 0.023\\ 
\bullet &\bullet & \bullet & \bullet & -0.046
 \end{smallmatrix}\right].
$$
\figurename~\ref{fig:LyapLivEx2} reports the level sets of  the function $W$ along with the solution to \eqref{eq:DT_plant} starting from $(0.3, 0.3)$. The picture clearly shows that the lack of symmetry of the nonlinearity $s$ reflects on the function $W$.
\begin{figure}
\centering
\psfrag{x1}[1][1][1]{$x_1$}
\psfrag{x2}[1][1][1]{$x_2$}
\includegraphics[trim=0cm 0cm 1cm 0cm, clip, width=0.95\columnwidth]{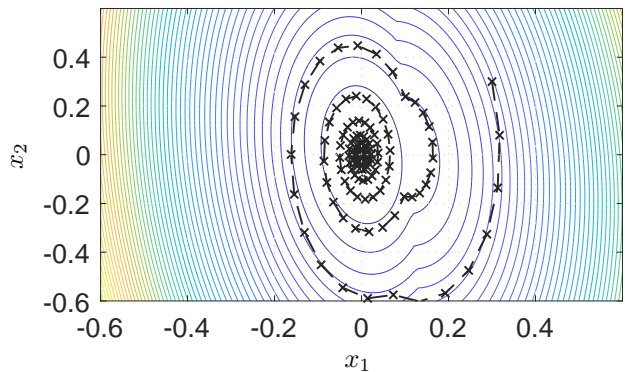}
\caption{Level sets of the function $W(x)=\sup V(x)$ in Example~\ref{ex:Ex2} and the trajectory $\phi$ of \eqref{eq:DT_plant} starting from $(0.3, 0.3)$ (dashed-crossed line).
\label{fig:LyapLivEx2}}
\end{figure}
\end{example}
\section{Conclusion}
The stability analysis of a class of discontinuous discrete-time control systems has been addressed. The proposed  approach relies on a characterization of the set-valued step mapping based on quadratic/linear constraints. Thanks to this characterization, we proposed a generalized quadratic set-valued Lyapunov function. Sufficient conditions in the form of LMIs have been provided to certify global exponential stability for the considered class of nonlinear control systems. The effectiveness of the methodology has been illustrated in two numerical examples, which have highlighted the potential of our approach in systematically generating generalized quadratic Lyapunov functions. 

The main thrust of our work is that it is unclear whether, for discrete-time systems, the use of nonquadratic Lyapunov functions for sector-bounded non-slope-restricted nonlinearities provides any advantage (for global stability). The approach  we propose has the major advantage to use a nonquadratic Lyapunov function to analyze systems without any slope restriction. Future directions of research include the extension to continuous-time control systems, as well as to regional stability analysis of systems with other discontinuous-nonlinearities. 
\bibliography{biblioQuant}  
\end{document}